\documentclass[a4paper,12pt]{amsart}

\usepackage[T1]{fontenc}
\usepackage{graphicx}
\usepackage{color}

\usepackage{amsmath,amssymb}
\usepackage{mathtools}
\newif\ifipco
\ipcofalse

\newif\ifshort
\shorttrue

\usepackage{subcaption}
\ifipco
\else
\fi
\usepackage{booktabs}
\usepackage{csquotes}
\usepackage{xspace}
\usepackage{svg}
\usepackage{multicol}
\usepackage{verbatim}   
\usepackage{ifthen}

\usepackage{pgf,tikz}
\usetikzlibrary{arrows}
\usetikzlibrary{patterns}
\usetikzlibrary[positioning]
\usetikzlibrary[fit]
\usetikzlibrary{shapes.geometric}
\usetikzlibrary{calc}

\newcommand\NN{{\mathbb N}}

\newcommand\RR{{\mathbb R}}
\newcommand\ZZ{{\mathbb Z}}
\newcommand\TT{{\mathbb T}}
\newcommand\cI{{\mathcal I}}

\newcommand\cN{{\mathcal N}}
\newcommand\cC{{\mathcal C}}
\newcommand\cS{{\mathcal S}}

\newcommand\cO{{\mathcal O}}
\newcommand\Cd{{d_c}}
\newcommand\Hd{{d_H}}

\newcommand\SetOf[2]{\left\{#1 \mid #2\right\}}
\newcommand\smallSetOf[2]{\{#1 \mid #2\}}
\newcommand\bigSetOf[2]{\bigl\{#1 \bigm| #2\bigr\}}
\newcommand\biggSetOf[2]{\biggl\{#1 \biggm| #2\biggr\}}
\newcommand\BiggSetOf[2]{\Biggl\{#1 \Biggm| #2\Biggr\}}

\DeclareMathOperator\cl{cl}
\DeclareMathOperator\conv{conv}

\DeclareMathOperator\Sym{Sym}
\DeclareMathOperator\argmax{arg\,max}

\DeclareMathOperator\supp{supp}

\newcommand\doi[1]{\href{http://dx.doi.org/#1}{\texttt{doi:#1}}}

\ifipco
\else
\usepackage[foot]{amsaddr}
\usepackage{fullpage}

\newtheorem{theorem}{Theorem}
\newtheorem{proposition}[theorem]{Proposition}
\newtheorem{corollary}[theorem]{Corollary}

\theoremstyle{remark}
\newtheorem{remark}[theorem]{Remark}

\newtheorem{example}[theorem]{Example}
\newtheorem{definition}[theorem]{Definition}
\newtheorem{question}[theorem]{Question}
\fi

\newcommand\polymake{\texttt{polymake}\xspace}

\newcommand\mptopcom{\texttt{mptopcom}\xspace}

\title{The Polyhedral Geometry of Truthful Auctions}

\ifipco
\author{Michael Joswig\inst{1,2} \and Max Klimm\inst{1} \and Sylvain Spitz\inst{1}} 

\institute{Technische Universität Berlin, 10623 Berlin, Germany \and
Max-Planck Institute for Mathematics in the Sciences, 04103 Leipzig, Germany}

\else
\author{Michael Joswig$^1$} 
\address{
  $^1$Technische Universit\"at Berlin,
  Discrete Mathematics/Geometry;
  Max-Planck Institute for Mathematics in the Sciences, Leipzig
}
\author{Max Klimm$^2$ \and Sylvain Spitz$^2$}
\address{
  $^2$Technische Universit\"at Berlin,
  Discrete Optimization
}
\email{$\{$joswig,klimm,spitz$\}$@math.tu-berlin.de}

\thanks{%
  Support by the Deutsche Forschungsgemeinschaft (DFG, German Research Foundation) under Germany's Excellence Strategy - The Berlin Mathematics Research Center MATH$^+$ (EXC-2046/1, project ID 390685689) gratefully acknowledged.
  M.~Joswig has further been supported by \enquote{Symbolic Tools in Mathematics and their Application} (TRR 195, project-ID 286237555).}
\fi

\begin{document}
\definecolor{cubecolor}{RGB}{97, 155, 242} 
\definecolor{triangcolor}{rgb}{0.4666666667 0.9254901961 0.6196078431}
\definecolor{spancolor}{RGB}{237, 113, 85}

\begin{abstract}
  The difference set of an outcome in an auction is the set of types that the auction mechanism maps to the outcome.
  We give a complete characterization of the geometry of the difference sets that can appear for a dominant strategy incentive compatible multi-unit auction showing that they correspond to regular subdivisions of the unit cube.
  This observation is then used to construct mechanisms that are robust in the sense that the set of items allocated to a player does change only slightly when the player's reported type is changed slightly.
\end{abstract}

\ifipco
\else
\maketitle
\fi


\section{Introduction}
Mechanism design is concerned with the implementation of favorable social outcomes in environments where information is distributed and only released strategically.
Specifically, this article is concerned with multi-dimensional mechanism design problems where a set of $m$ items is to be allocated to a set of $n$ players. 
The attitude of each player for receiving a subset of the items is determined by the so-called \emph{type} of the player and is their private information and not available to the mechanism. In this setting, a mechanism elicits the types from the players, and---based on the reported types---decides on an allocation of the items to the players, and on a price vector that specifies the amount of money that the different players have to pay to the mechanism.
In order to incentivize the players to truthfully report their true types to the mechanism, one is interested in mechanisms that have the property that no matter what the other players report to the mechanism, no player can benefit from misreporting their type; mechanisms that enjoy this property are called \emph{dominant strategy incentive compatible}, short DSIC.
In this paper, we investigate the geometric properties of DSIC mechanisms. Because DSIC mechanisms require truthful reporting of the type no matter of the types declared by the other players, they can be characterized by the one-player mechanisms that arise when the declared valuations of the other players are fixed.  
 
As an example for a mechanism, consider the basic case of a combinatorial auction (see De Vries and Vohra~\cite{deVriesVohra2003} for surveys) where two items are sold to two players with additive valuations.
In that case each player~$i$ has a two-parameter type $\theta_i = (\theta_{i,1}, \theta_{i,2})$ where the scalar $\theta_{i,j}$ is the monetary equivalent that player~$i$ attaches to receiving item~$j$. For illustration, assume that player~$2$ reported $\theta_2' = (1,1)$ and consider the corresponding one-player mechanism for player~$1$. If each item~$j$ is sold independently to the bidder~$i$ with the highest reported type $\theta_{i,j}'$ (breaking ties in favor of player~$1$), we obtain that player~$1$ receives item $j$ if and only of $\theta_{1,j}' \geq 1$.
Geometrically, this one-player mechanism can be represented by its \emph{difference sets} $\smash{Q_S}$, $\smash{S \in 2^{\{1,2\}}}$ where $\smash{Q_S}$ is equal to the closure of the set of types reported by player~$1$ so that they get allocated the set of items $S$. The difference sets were introduced by Vohra~\cite[p.~41]{Vohra2011} and reveal valuable information about the properties of the mechanism; see Fig.~\ref{fig:difference-sets}. Fig.~\ref{fig:difference-sets1} shows the difference sets of the mechanism selling each item to the highest bidder; Fig.~\ref{fig:difference-sets2} and Fig.~\ref{fig:difference-sets3} show the difference sets of other DSIC mechanisms (not specified here). 

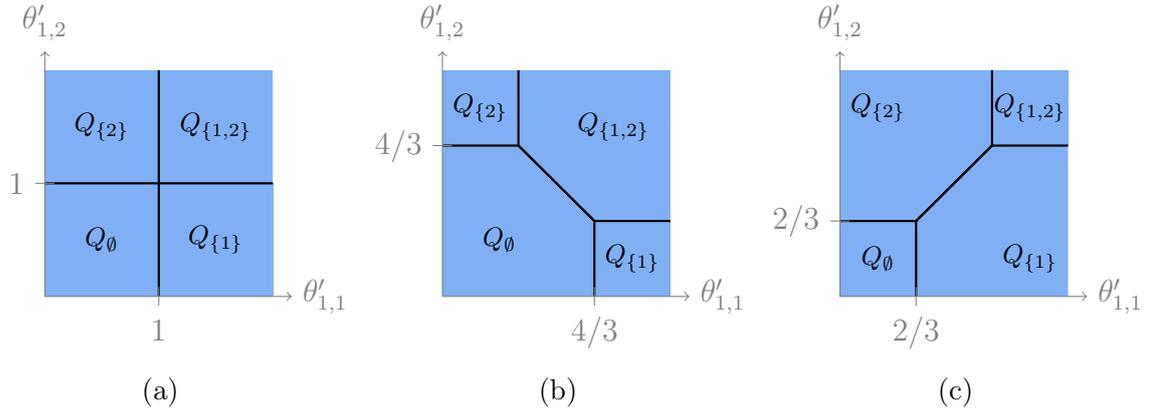
\begin{figure}[tb]
  \center 
  \begin{subfigure}{0.32\textwidth}
    \center
    \ifipco
    \begin{tikzpicture}[scale = .42]
    \else
    \begin{tikzpicture}[scale = .5]
    \fi
      \tikzstyle{facet} = [color=cubecolor, opacity=0.8];
      \tikzstyle{edge} = [thick];
      \tikzstyle{axis} = [->, color=gray];
   
      \coordinate (0) at (0,0);
      \coordinate (1) at (0,3);
      \coordinate (2) at (0,6);
      \coordinate (3) at (3,3);
      \coordinate (4) at (3,6);
      \coordinate (5) at (3,0);
      \coordinate (6) at (3,3);
      \coordinate (7) at (6,0);
      \coordinate (8) at (6,3);
      \coordinate (9) at (6,6);
      \fill[facet] (0)--(1)--(3)--(6)--(5)--(0)--cycle;
      \fill[facet] (1)--(2)--(4)--(3)--(1)--cycle;
      \fill[facet] (3)--(4)--(9)--(8)--(6)--(3)--cycle;
      \fill[facet] (5)--(6)--(8)--(7)--(5)--cycle;
      \draw[edge] (1)--(3)--(6)--(5);
      \draw[edge] (4)--(3)--(6)--(8);
      \draw[axis] (0,0) -- (6.5,0) node[right] {$\theta_{1,1}'$};
      \draw[axis] (0,0) -- (0,6.5) node[above] {$\theta_{1,2}'$};
      \draw[gray] (3,0.25) -- (3,-0.25) node[below] {\textcolor{gray}{$1\vphantom{/1}$}};
      \draw[gray] (0.25,3) -- (-0.25,3) node[left] {\textcolor{gray}{$\hphantom{1/}1$}};
   
      \node (r11) at (1.5,4.5) {\footnotesize $Q_{\{2\}}$};
      \node (r11) at (4.5,4.5) {\footnotesize $Q_{\{1,2\}}$};
      \node (r11) at (4.5,1.5) {\footnotesize $Q_{\{1\}}$};
      \node (r11) at (1.5,1.5) {\footnotesize $Q_{\emptyset}$};
    \end{tikzpicture}
    \caption{} \label{fig:difference-sets1}
  \end{subfigure}
  \begin{subfigure}{0.32\textwidth}
    \center
    \ifipco
    \begin{tikzpicture}[scale = .42]
    \else
    \begin{tikzpicture}[scale = .5]
    \fi
      \tikzstyle{facet} = [color=cubecolor, opacity=0.8];
      \tikzstyle{edge} = [thick];
      \tikzstyle{axis} = [->, color=gray];
    
      \coordinate (0) at (0,0);
      \coordinate (1) at (0,4);
      \coordinate (2) at (0,6);
      \coordinate (3) at (2,4);
      \coordinate (4) at (2,6);
      \coordinate (5) at (4,0);
      \coordinate (6) at (4,2);
      \coordinate (7) at (6,0);
      \coordinate (8) at (6,2);
      \coordinate (9) at (6,6);
      \fill[facet] (0)--(1)--(3)--(6)--(5)--(0)--cycle;
      \fill[facet] (1)--(2)--(4)--(3)--(1)--cycle;
      \fill[facet] (3)--(4)--(9)--(8)--(6)--(3)--cycle;
      \fill[facet] (5)--(6)--(8)--(7)--(5)--cycle;
      \draw[edge] (1)--(3)--(6)--(5);
      \draw[edge] (4)--(3)--(6)--(8);
      \draw[axis] (0,0) -- (6.5,0) node[right] {$\theta_{1,1}'$};
      \draw[axis] (0,0) -- (0,6.5) node[above] {$\theta_{1,2}'$};
      \draw[gray] (4,0.25) -- (4,-0.25) node[below] {\textcolor{gray}{$4/3$}};
      \draw[gray] (0.25,4) -- (-0.25,4) node[left] {\textcolor{gray}{$4/3$}};
    
      \node (r11) at (1,5) {\footnotesize $Q_{\{2\}}$};
      \node (r11) at (4.5,4.5) {\footnotesize $Q_{\{1,2\}}$};
      \node (r11) at (5,1) {\footnotesize $Q_{\{1\}}$};
      \node (r11) at (1.5,1.5) {\footnotesize $Q_{\emptyset}$};
    \end{tikzpicture}
    \caption{} \label{fig:difference-sets2}
  \end{subfigure}
  \begin{subfigure}{0.32\textwidth}
    \center
    \ifipco
    \begin{tikzpicture}[scale = .42]
    \else
    \begin{tikzpicture}[scale = .5]
    \fi
      \tikzstyle{facet} = [color=cubecolor, opacity=0.8];
      \tikzstyle{edge} = [thick];
      \tikzstyle{axis} = [->, color=gray];
  
      \coordinate (0) at (0,0);
      \coordinate (1) at (0,2);
      \coordinate (2) at (0,6);
      \coordinate (3) at (2,2);
      \coordinate (4) at (4,6);
      \coordinate (5) at (2,0);
      \coordinate (6) at (4,4);
      \coordinate (7) at (6,0);
      \coordinate (8) at (6,4);
      \coordinate (9) at (6,6);
      \fill[facet] (0)--(1)--(3)--(5)--(0)--cycle; 
      \fill[facet] (1)--(2)--(4)--(6)--(3)--(1)--cycle; 
      \fill[facet] (4)--(9)--(8)--(6)--cycle; 
      \fill[facet] (5)--(3)--(6)--(8)--(7)--(5)--cycle; 
      \draw[edge] (1)--(3)--(6)--(8);
      \draw[edge] (4)--(6)--(3)--(5);
      \draw[axis] (0,0) -- (6.5,0) node[right] {$\theta_{1,1}'$};
      \draw[axis] (0,0) -- (0,6.5) node[above] {$\theta_{1,2}'$};
      \draw[gray] (2,0.25) -- (2,-0.25) node[below] {\textcolor{gray}{$2/3$}};
      \draw[gray] (0.25,2) -- (-0.25,2) node[left] {\textcolor{gray}{$2/3$}};
  
      \node (r11) at (1,5) {\footnotesize $Q_{\{2\}}$};
      \node (r11) at (5,5) {\footnotesize $Q_{\{1,2\}}$};
      \node (r11) at (5,1) {\footnotesize $Q_{\{1\}}$};
      \node (r11) at (1,1) {\footnotesize $Q_{\emptyset}$};
    \end{tikzpicture}
    \caption{} \label{fig:difference-sets3}
  \end{subfigure}
  \caption{
    \label{fig:difference-sets}
    Difference sets of several mechanisms; cf.\ \cite[Fig.~1]{Vidali09}.}
\end{figure}

Under reasonable assumptions, the difference sets form a polyhedral decomposition of the type space. In this paper, we are interested in characterizing their polyhedral geometry.
%
This is a continuation of work of Vidali~\cite{Vidali09} who showed that the two combinatorial types shown in Fig.~\ref{fig:difference-sets2}~and~\ref{fig:difference-sets3} are the only cases that can appear for a DSIC mechanism for two items (where the combinatorial type in Fig.~\ref{fig:difference-sets1} is a common degenerate case of both).
She then also provided a similar characterization of the combinatorial types that can appear for three items and asked how these findings can be generalized to more items.

Characterizing the combinatorial types of mechanisms is interesting for a variety of reasons. First, such geometric arguments are often used in order to characterize the set of allocation functions that are implementable by a DSIC mechanism. For instance, the difference sets whose closures have nonempty intersection correspond exactly to two-cycles in an auxiliary network used by Rochet~\cite{Rochet87} in order to characterize the allocation functions that are implementable by DSIC mechanisms.
Second, the combinatorial types can be used to study the sensitivity of mechanisms to deviations in the reported types. As an example consider the mechanism in Fig.~\ref{fig:difference-sets1}. For any $\epsilon > 0$, reporting the type $(1-\epsilon, 1-\epsilon)$ yields no item for player~$1$ while the report of the type $(1+\epsilon, 1+\epsilon)$ grants them both items. Put differently, a small change in the reported type may change the outcome from no items being allocated to player~$1$ to all items being allocated to the same player. This is in contrast with the mechanism shown in Fig.~\ref{fig:difference-sets3} where a small change in the reported size may change the cardinality of the set of allocated items only by $1$.
Third, the combinatorial types of the mechanism are relevant for the efficiency of the mechanism; see, e.g.,~\cite{Christodoulou08,Vidali09}.

\newcommand\ourresultssec{Our results}
\ifipco
\else
\subsection{\ourresultssec}
\fi
We give a complete characterization of the combinatorial types of all DSIC combinatorial auctions with $m$ items for any value of $m$. This answers of Vidali~\cite{Vidali09} on how to generalize her results for $m =2$ and $m=3$ to larger values of $m$.
We employ methods from polyhedral geometry \cite{Triangulations} and tropical combinatorics \cite{ETC}.
Our results rest on the observation that a multi-player mechanism is DSIC if and only if its single-player components are DSIC; see \cite{SaksYu05}.
We show that for $m$ items, the combinatorial types of those single-player components are in bijection to equivalence classes of regular subdivisions of the $m$-dimensional unit cube.
We identify the relevant symmetries for exchangeable items and conclude that there are exactly 23 nondegenerate combinatorial types for $m=3$ and 3{,}706{,}261 such types for $m=4$ (Theorem~\ref{thm:TriangCube}).
We then use this characterization to study the optimal sensitivity of mechanisms to slight changes in the reported types. Specifically, we show that for any number of items $m$, there is a one-player combinatorial auction so that the cardinality of the set of items received by the player changes by at most $1$ when the reported type is slightly perturbed (Proposition~\ref{prop:CardStab}). We also give bounds on a similar measure involving the Hamming distance of the set of received items (Proposition~\ref{pro:hamming}).  
\ifshort
In the full version of this paper 
we further show how to apply the same methodology in order to classify the combinatorial types of affine maximizers with $n$ players.
\fi


\newcommand\furtherrelatedsec{Further related work}
\ifipco
\else
\subsection{\furtherrelatedsec}
\fi

Rochet's Theorem~\cite{Rochet87} states that an allocation function is implementable by a DSIC mechanism if and only if the allocation networks of the corresponding one-player mechanisms have no finite cycles of negative lengths. There is a substantial stream of literature exhibiting conditions where it is enough to require conditions on shorter cycles \cite{Archer2014,Ashlagi2010,Berger2017,Bikhchandani2006,Carbajal2015,Edelman2021,Kushnir2021}; for instance, it suffices to require the nonnegativity for cycles of length $2$ when the preferences are single-peaked \cite{Mishra2014} or when the type-space is convex \cite{SaksYu05}.
Roberts~\cite{Roberts} showed that when the type space of all players is $\RR^{\Omega}$, then only affine maximizers are implementable by a DSIC mechanism.
Gui et al.~\cite{Gui2004} and Vohra~\cite{Vohra2011} studied the difference sets $Q_A$ of a mechanism and showed that under reasonable assumptions their closures are polyhedra. Vidali~\cite{Vidali09} studied the geometry of the polyhedra for the case of two and three items.

In recent years, tropical geometric methods proved to be useful for algorithmic game theory and lead to new results in mechanism design \cite{BaldwinKlemperer,NgocLin19,NgocYu19}, mean payoff games \cite{MeanPayoff09,MeanPayoff14}, linear optimization \cite{ABGJ:SIREV} and beyond.
Beyond the scope of combinatorial auctions, our results are also applicable to the mechanism design problem of scheduling on unrelated machines \cite{Christo22,ChristodoulouEtAl09,DobzinskiS20,Giannakopoulos2021,NisanRonen}.

\section{Preliminaries}
\label{sec:prelim}

In this section, we give a brief overview of basic concepts from mechanism design theory, polyhedral geometry and tropical combinatorics used in this paper.
For a more comprehensive treatment we refer to \cite{GameTheory}, \cite{Triangulations} and \cite{ETC,Tropical+Book}.

\newcommand\mechanismdesignsec{Mechanism design}
\ifipco
\else
\subsection{\mechanismdesignsec}
\fi
A \emph{multi-dimensional mechanism design problem} consists of a finite set $[m] \coloneqq \{1, \dots, m\}$ of \emph{items} and a finite set $[n]\coloneqq \{1, \dots, n\}$ of players. 
Every player~$i$ has a set of possible \emph{types} $\Theta_i \subseteq \mathbb{R}^m$ where for $\theta_i = (\theta_{i,1},\dots,\theta_{i,m}) \in \Theta_i$ the value $\theta_{i,j}$, $j \in [m]$ is the monetary value player~$i$ attaches to the fact of receiving item~$j$. 
A vector $\theta = (\theta_1,\dots,\theta_n)$ with $\theta_i \in \Theta_i$ for all $i \in [n]$ is called a \emph{type vector} and $\Theta = \Theta_1 \times \dots \times \Theta_n$ is the space of all type vectors. The type~$\theta_i$ is the private information of player~$i$ and unknown to all other players $j \neq i$ and the mechanism designer.
Let 
\[
  \Omega \ = \ \Bigg\{A \in \{0,1\}^{n \times m} \; \Bigg\vert \; \sum_{i \in [n]} A_{i,j}=1 \text{ for all } j \in [m]\Bigg\}
\]
be the set of allocations of the $m$ items to the $n$ players.
Here, the $i$-th row $A_i$ of an allocation matrix $A \in \Omega$ corresponds to the allocation for the $i$-th player.
A \emph{(direct revelation) mechanism} is a tuple $M = (f,p)$ consisting of an \emph{allocation function} \ifipco \else \linebreak \fi $f : \Theta \to \Omega$ and a \emph{payment function} $p : \Theta \to \RR^n$.
The mechanism first elicits a claimed type vector $\theta' = (\theta'_1,\dots,\theta_n') \in \Theta$ where $\theta_i' \in \Theta_i$ is the type reported by player~$i$. It then chooses an alternative $f(\theta')$ and payments $p(\theta')  = (p_1(\theta'),\dots,p_n(\theta')) \in \RR^n$ where $p_i(\theta')$ is the payment from player~$i$ to the mechanism.
We assume that the players' utilities are \emph{quasi-linear} and the valuation is \emph{additive}, i.e., the utility of player~$i$ with type $\theta_i$ when the type vector reported to the mechanism is $\theta'$ is
\ifipco
\else
\begin{align*}
u_i(\theta' \mid \theta_i) \ = \ f_i(\theta') \cdot \theta_i - p_i(\theta')\enspace ,
\end{align*}
\fi
where $f_i(\theta') \in \{0,1\}^m$ is the characteristic vector of the items allocated to player~$i$ when the reported type vector is $\theta'$ and $f_i(\theta')\cdot \theta_i$ denotes the scalar product.
A direct revelation mechanism is called \emph{dominant strategy incentive compatible} (DSIC) or \emph{truthful} if 
\ifipco
\else
\begin{align}
\label{eq:ic}\tag{IC}
u_i(\theta \mid \theta_i) \ \geq \ u_i((\theta_i', \theta_{-i}) \mid \theta_i)\enspace,
\end{align}
\fi
for all $i \in [n]$, $\theta \in \Theta$,  and $\theta' \in \Theta_i$. Here and throughout, $(\theta_i', \theta_{-i})$ denotes the type vector where player~$i$ reports~$\theta_i'$ and every other player~$j$ reports $\theta_j$ as in $ \theta$.

For an allocation $A \in \Omega$, let $R_{A}~=~\{\theta \in \Theta \mid  f(\theta) = A\}$ be the preimage of $A$ under $f$, and let $Q_{A} = \cl(R_{A})$ be the topological closure of $R_{A}$.   
We call $Q_{A}$ the \emph{difference set} of $A$.

%
%
%
\newcommand\tropicalsec{Polyhedral geometry and tropical combinatorics}
\ifipco
\else
\subsection{\tropicalsec}
\fi
\ifipco\else Consider \fi the \emph{max-tropical semiring} $(\TT, \oplus, \odot)$ with $\TT \coloneqq \RR \cup \{-\infty\}$, $a \oplus b \coloneqq \max\{a, b\}$ and $a \odot b \coloneqq a + b$.
Picking coefficients $\lambda_u\in\TT$ for $u\in\ZZ^m$ such that only finitely many are distinct from $-\infty$ defines an $m$-variate tropical (Laurent) polynomial, $p$, whose evaluation at $x \in \RR^m$ reads
\begin{equation}\label{eq:trop+poly}
 p(x) \ = \ \bigoplus_{u \in \ZZ^m} \lambda_u \odot x^{\odot u} \ = \ \max\SetOf{ \lambda_u + x \cdot u}{u \in \ZZ^m} \enspace .
\end{equation}
The \emph{support} of $p$ is the set $\supp(p)=\SetOf{u \in \ZZ^m}{\lambda_u \neq -\infty}$.
The \emph{tropical hypersurface} $V(p)$ is the set of points $x \in \RR^m$ such that the maximum in \eqref{eq:trop+poly} is attained at least twice.
The tropical hypersurface partitions the set $\RR^m\setminus V(p)$ into sets in which the maximum of \eqref{eq:trop+poly} is attained exactly once, for some fixed $u\in\ZZ^m$; see Fig.~\ref{fig:tropVar}.
Taking the closure of such a part, we get a \emph{region} of $V(p)$, which is a (possibly unbounded) polyhedron.
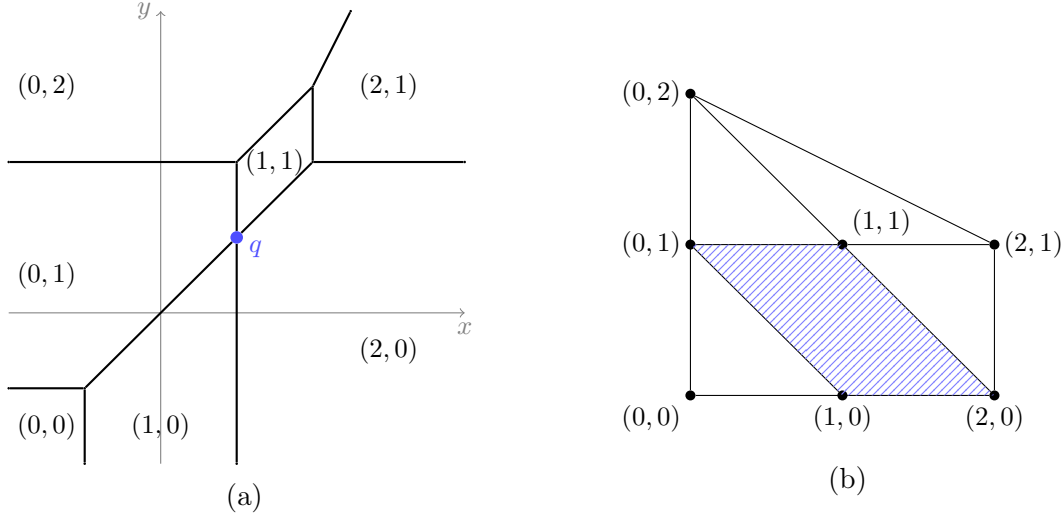
\begin{figure}[tb]
  \center
  \begin{subfigure}{0.49\textwidth}
    \center
    \ifipco
    \begin{tikzpicture}[scale=0.85]
      \scriptsize
    \else
    \begin{tikzpicture}
      \footnotesize
    \fi
      \tikzstyle{axis} = [->, color=gray];
      \tikzstyle{trop} = [thick];
      \tikzstyle{u} = [circle, scale=0.1, fill=black];
      \draw[axis] (-2,0) -- (4,0);
      \draw[axis] (0,-2) -- (0,4);
      \node[below, color=gray] at (4,0) {$x$};
      \node[left, color=gray] at (0,4) {$y$};
      
      \node (1) at (0, 0) {};
      \node[u] (2) at (1, 2) {};
      \node[u] (3) at (-1, -1) {};
      \node[circle, scale=0.5, fill=blue!70] (4) at (1, 1) {};
      \node[u] (5) at (2, 2) {};
      \node[u] (6) at (2, 3) {};
      \node[u] (b1) at (1,-2) {};
      \node[u] (b2) at (-1,-2) {};
      \node[u] (b3) at (-2,-1) {};
      \node[u] (b4) at (-2,2) {};
      \node[u] (b5) at (4,2) {};
      \node[u] (b6) at (2.5,4) {};
      
      \node at (-1.5, -1.5) {$(0,0)$};
      \node at (-0, -1.5) {$(1,0)$};
      \node at (-1.5, 0.5) {$(0,1)$};
      \node at (3, -.5) {$(2,0)$};
      \node at (-1.5, 3) {$(0,2)$};
      \node at (1.5, 2) {$(1,1)$};
      \node at (3, 3) {$(2,1)$};
      \node[color=blue!70] at (1.25,0.85) {$q$};
    
      \draw[trop] (b2)--(3)--(4)--(2)--(6)--(5)--(4)--(b1);
      \draw[trop] (b3)--(3);
      \draw[trop] (b4)--(2);
      \draw[trop] (b5)--(5);
      \draw[trop] (b6)--(6);
    \end{tikzpicture}
    \caption{} \label{fig:tropVar}
  \end{subfigure}
  \begin{subfigure}{0.49\textwidth}
    \center
	  \vspace{.8cm}
    \ifipco
    \begin{tikzpicture}[scale = 1.7]
      \scriptsize	
    \else
    \begin{tikzpicture}[scale = 2]
      \footnotesize
    \fi
      \tikzstyle{u} = [circle, scale=0.4, fill=black];
      \tikzstyle{v} = [];
      \node[u] (0) at (0,0) {};
      \node[u] (1) at (1,0) {};
      \node[u] (2) at (0,1) {};
      \node[u] (3) at (2,0) {};
      \node[u] (4) at (1,1) {};
      \node[u] (5) at (0,2) {};
      \node[u] (6) at (2,1) {};
      
      \node[below left] at (0,0) {$(0,0)$};
      \node[below] at (1,0) {$(1,0)$};
      \node[left] at (0,1) {$(0,1)$};
      \node[below] at (2,0) {$(2,0)$};
      \node[left] at (0,2) {$(0,2)$};
      \node[above right] at (1,1) {$(1,1)$};
      \node[right] at (2,1) {$(2,1)$};
      
      \draw[v] (0)--(3)--(6)--(5)--(0);
      \draw[v] (1)--(2)--(4)--(3);
      \draw[v] (5)--(4)--(6);
      
      \fill[pattern=north east lines, pattern color=blue!50] (0,1)--(1,0)--(2,0)--(1,1)--cycle;
	  \end{tikzpicture}
	  \vspace{.1cm}
    \caption{}
  \end{subfigure}
  \caption{(a) Tropical hypersurface $V(p)$ and (b) dual regular subdivision of $\supp(p)$, where $p(x, y) = \max \{0, x+1, y+1, 2x, x+y, 2y-1, 2x+y -2\}$.
   In (a) regions are marked by their support vectors; in (b) the same labels mark vertices of the subdivision.
   Conversely, e.g., the blue quadrangular cell on the right is dual to the vertex $q$ on the left.} \label{fig:tropDual}
\end{figure}
The \emph{Newton polytope} $\cN(p)=\conv(\supp(p))$ of a tropical polynomial~$p$ is the convex hull of its support.
\ifipco
\else
In the present work tropical hypersurfaces often occur in the guise of their dual regular subdivisions.
This requires some more details.
\fi
A finite set $\cC$ of polyhedra in $\RR^m$ is a \emph{polyhedral complex}, if it is closed with respect to taking faces and if for any two $P,Q \in \cC$ the intersection $P \cap Q$ is a face of both $P$ and $Q$.
\ifipco
\else
Note that the empty set is considered to be a face of each polyhedron.
If the maximal polyhedra in $\cC$ happen to share the same dimension, the complex $\cC$ is called \emph{pure}.
\fi
The polyhedra in $\cC$ are the \emph{cells} of $\cC$.
If every polyhedron in $\cC$ is bounded, it is a \emph{polytopal complex}.
Further, given a finite set of points $U \subset \RR^m$ and a polytopal complex $\cC$ in $\RR^m$, $\cC$ is a \emph{polytopal subdivision} of $U$ if the vertices of all polytopes in $\cC$ are points in $U$ and if the union of all polytopes in $\cC$ is the convex hull of the points in $U$.
If the cells of a polytopal subdivision are all simplices, it is called a \emph{triangulation}.
\ifipco
\else
\begin{definition}[Regular Subdivision] \label{def:RegSubd}
 Let $U \subset \RR^m$ be a finite set of points and let $\lambda:U \rightarrow \RR$ be a lifting function on $U$. Consider the \emph{lifted polytope}
 \[
  P(U,\lambda) \ \coloneqq \ \conv \biggSetOf{(a, \lambda(a)) \in \RR^{m+1}}{a \in U} \enspace .
 \] 
 Its \emph{upper faces} have an outer normal vector with positive last coordinate.
 Projecting these upper faces by omitting the last coordinate yields a polytopal subdivision of $U$.
 This is called the \emph{regular subdivision} of $U$ induced by $\lambda$.
\end{definition}
\fi
The following proposition explains the duality between a tropical hypersurface and the regular subdivision of its support.
Here we identify a polytopal subdivision with its finite set of cells, partially ordered by inclusion.
A proof can be found in \cite[Theorem 1.13.]{ETC}; see Fig.~\ref{fig:tropVar} for an example.

\begin{proposition} \label{prop:TropDual}
 Let $p = \max\SetOf{ \lambda_u + x \cdot u}{u \in \ZZ^m}$ be a tropical Laurent polynomial.
 Then there is an inclusion reversing bijection between the regular subdivision of the support of $p$ with respect to $\lambda(u) = \lambda_u$ and the polyhedral complex induced by the regions of $V(p)$.
\end{proposition}

\ifipco
\else
The inclusion reversion of the bijection can be seen in the following way.
The bijection maps a vertex $u$ of $\cN(p)$ to the region, where $u$ is the support vector maximizing $p$.
Further, if the intersection of a tuple of regions is nonempty, it gets mapped to the convex hull of the corresponding support vectors in the Newton polytope.
An example can be seen in Fig.~\ref{fig:tropDual}.
\fi

\section{Characterization of One-Player Mechanisms}
\label{sec:CombsOfAlls}

In many applications, the intersections of difference sets $Q_A$ are special, as the mechanism is essentially indifferent between the outcomes and uses a tie-breaking rule or random selection to determine the outcome.
Observing which difference sets intersect and which do not gives rise to a combinatorial pattern which we attribute to the allocation function.
We want to study these patterns in order to classify which of them can be attributed to truthfulness. 
We express such a pattern as an abstract simplicial complex over the allocation space and call it the \emph{indifference complex} of the allocation function.

Formally, an \emph{abstract simplicial complex} over some finite set $E$ is a nonempty set family $\cS$ of subsets of $E$, such that for any set $S \in \cS$ and any subset $T \subseteq S$, we also have $T \in \cS$.
The elements of an abstract simplicial complex are called \emph{faces}.
The dimension of $\cS$ is the maximal cardinality of any face, minus one.

\begin{definition}[Indifference Complex]
	The \emph{indifference complex} $\cI(f)$ of an allocation function $f$ is the abstract simplicial complex defined as
	\[
		\cI(f) \ = \ \bigSetOf{\cO \subseteq \Omega}{\bigcap_{A \in \cO} Q_{A} \neq \emptyset} \enspace .
	\]
\end{definition}
Note that the indifference complex $\cI(f)$ is precisely the nerve complex of the family of difference sets of $f$; see \cite[\S10]{Bjorner:1995}.
We call an allocation function $f$ \emph{implementable} if there is a DSIC mechanism $M = (f,p)$ and we call an indifference complex $\cI$  implementable, if there is an implementable allocation function $f$ such that $\cI(f) = \cI$.

We define the \emph{local allocation function} of player $i$ for a given type vector $\theta_{-i}$ to be $f_{i,\theta_{-i}}(\theta_i) = A_i$, where $A_i$ is the $i$-th row of $A = f(\theta_i, \theta_{-i})$.
Further, let us fix a payment vector $p \in \RR^{2^m}$, which we index by allocations $a \in \{0,1\}^m$.
Then, for any type $\theta_i \in \RR^m$, we let 
\ifipco
\else
\[
  u_{p}(\theta_i) \ = \ \max\biggSetOf{\theta_i \cdot a - p_a}{a \in \{0,1\}^m} \enspace .
\]
\fi
The resulting function $u_p:\RR^m\to\RR$ is a max-tropical polynomial of degree $m$.
We refer to the allocation in $\{0,1\}^m$ which maximizes $u_p(\theta_i)$ as $\argmax u_p(\theta_i)$.
We restate \cite[Proposition 9.27]{GameTheory}, which says that in a truthful setting, the local allocation functions are defined by such tropical polynomials, where the vector $p$ depends only on the types of the other players.
\begin{proposition} \label{prop:TruthTrop}
  The allocation function $f$ is truthful, if and only if for all players $i \in [n]$ and all type vectors $\theta \in \Theta$, there exists a payment vector $p_i(\theta_{-i}) \in \RR^{2^m}$, such that $f_{i,\theta_{-i}}(\theta_i) \in \argmax u_{p_i(\theta_{-i})}(\theta_i)$.
\end{proposition}

An important consequence of Proposition~\ref{prop:TruthTrop} is that an allocation function $f$ is truthful if and only if all of its local function $f_{i,\theta_{-i}}$ are truthful.
Therefore, for the remainder of this section, we fix a player~$i$ and the type vector $\theta_{-i}$ of the other players and consider the corresponding one-player mechanism for player~$i$.
Equivalently, we assume that $n=1$ and that the allocation function $f$ chooses a single-player allocation $a \in \Omega = \{0,1\}^m$.

Next, we want to point out the relationship between the indifference complex and the \emph{allocation network}, which is a tool often used to analyze the truthfulness of allocation functions.
It is the weighted complete directed graph $G_f$ with a node for each allocation $a \in \Omega$ and with arc lengths 
\ifipco
\else
	\begin{align*}
    \ell(a, a') \ = \ \inf_{\theta \in \Theta : \theta \in R_{a'} } \{\theta \cdot a' - \theta \cdot a	\}\enspace .
	\end{align*}
\fi
The arc length $\ell(a, a')$ is the minimal loss of the player's valuation that would occur when the mechanism changes from allocation $a'$ to $a$, while having a type in the difference set $Q_{a'}$.

We can link the indifference complex and the allocation network through the following proposition.
It is a generalization of \cite[Proposition 5]{SaksYu05} and was given in \cite[Lemma 3]{Vidali09} without a proof, we restate it here in our notation, adding a short proof.

\begin{proposition}\label{prop:allocCycles}
 Let $(f,p)$ be a DSIC mechanism with quasi-linear utilities for one player. Let $C = (a^{(1)}, \dots, a^{(k)} = a^{(1)})$ be a cycle in the allocation network $G_f$, such that for each $j \in [k-1]$, we get $Q_{a^{(j)}} \cap Q_{a^{(j+1)}} \neq \emptyset$.
 Then the length of the cycle $C$ is $0$.
\end{proposition}
\begin{proof}
 Let $C = (a^{(1)}, \dots, a^{(k)})$ be a cycle as in the statement of the proposition.
 \ifipco \else \linebreak \fi Using \cite[Proposition 5]{SaksYu05}, we obtain that for any $\theta \in Q_{a^{(j)}} \cap Q_{a^{(j+1)}}$, the equation $\theta \cdot a^{(j)} - \theta \cdot a^{(j+1)} = \ell(a^{(j)},a^{(j+1)})$ is satisfied.
 Since the mechanism $(f,p)$ is truthful and $\theta \in Q_{a^{(j)}} \cap Q_{a^{(j+1)}}$, we get $\theta \cdot a^{(j)} - p_{a^{(j)}} = \theta \cdot a^{(j+1)} - p_{a^{(j+1)}}$.
 Therefore $p_{a^{(j)}} - p_{a^{(j+1)}} = \ell(a^{(j)},a^{(j+1)})$.
 Adding up all the lengths of the arcs of $C$ we get 0, which finishes the proof.
\end{proof}

A consequence of Proposition~\ref{prop:allocCycles} is that all cycles in $G_f$ with the property that all of its edges connect two common nodes of some face $\cO \subseteq \Omega$ of the indifference complex $\cI(f)$, have length $0$.
Especially, each oriented cycle in the one-skeleton of $\cI(f)$ is also a zero-cycle in $G_f$. 

For the remainder of this section, our goal is to classify truthful allocation functions for the given type of allocation mechanisms.
Recall that Proposition~\ref{prop:TruthTrop} shows that the difference sets of a truthful one-player allocation mechanism are exactly the regions of the tropical utility function of the player.
This is the key observation we use to prove our first main result.
It states that there is a bijection between implementable one-player indifference complexes for $m$ items and the regular subdivisions of the $m$-dimensional cube.

\begin{theorem} \label{thm:TriangCube}
  An indifference complex $\cI$ for $m$ items and one player is implementable if and only if there is a regular subdivision $\cS$ of the $m$-dimensional cube, such that the facets of $\cI$ are precisely the vertex sets of the maximal cells of $\cS$.
\end{theorem}
\begin{proof}
 Let $\cI$ be an indifference complex.
 It is implementable if and only if there exists a truthful allocation function $f$ with $\cI(f) = \cI$.
 By Proposition~\ref{prop:TruthTrop}, this is equivalent to the fact that there is a payment vector $p \in \RR^{2^m}$, such that the difference sets $Q_a$ are exactly the regions of the tropical hypersurface $V(u_p)$.
 As the Newton polytope of $u_p$ is the unit cube $[0,1]^m$, Proposition~\ref{prop:TropDual} provides a duality between the difference sets $Q_a$ and the regular subdivision of $[0,1]^m$ with respect to the payment vector $p$.
 Hence, a maximal cell in the regular subdivision with vertices $(a^{(1)}, \dots, a^{(k)})$ corresponds to a maximal set of allocations such that $\bigcap_{a \in \{a^{(1)} ,\dots, a^{(k)}\}} Q_a \neq \emptyset$.
 The latter is a facet of $\cI(f)$.
\end{proof}

Note that the proof is constructive.
Further, Theorem~\ref{thm:TriangCube} says that the simplicial complex $\cI(f)$ is precisely the crosscut complex of the poset of cells of the regular subdivision $\cS$ \cite[\S10]{Bjorner:1995}.
If $\cS$ is a triangulation then its crosscut complex is $\cS$ itself, seen as an abstract simplicial complex.
The main consequence of Theorem~\ref{thm:TriangCube} is that for truthful one-player mechanisms with additive and quasi-linear utilities, the partitioning of the type space into difference sets is characterized by the duality to the regular subdivision of the cube, which is captured by the indifference complex. 

\begin{example} \label{ex:counter}
 Consider the case $m=3$ and $n=1$ where the player has a type $\theta \in \RR^3$.
 Let $f$ be the local allocation function defined as
 \ifipco
\else
 \[
  f(\theta) \ \in \ \argmax\SetOf{\theta \cdot a - p_a}{a \in \{0,1\}^3}, 
 \]
 \fi
 with $p_{000} = 0$, $p_{100} = p_{010} = p_{001} = 1/4$, $p_{110} = p_{101} = p_{011} = 2/3$, and $p_{111} = 5/6$.
 Fig.~\ref{fig:counter} shows the type space for $\theta \in [0,1]^3$.
 The five maximal cells of $\cI(f)$ are
 \ifipco
\else
 \begin{multline*}
 \{000,100,010,001\}, \{100,010,110,111\}, \{100,001,101,111\}, \\ 
 \{010,001,011,111\} \text{ and } \{110,101,011,111\} \enspace .
 \end{multline*}
 \fi
 These cells form a regular triangulation of $[0,1]^3$, which is type F in Fig.~\ref{fig:AllocsAndTriangs}.
\end{example}

\begin{figure}[tb]
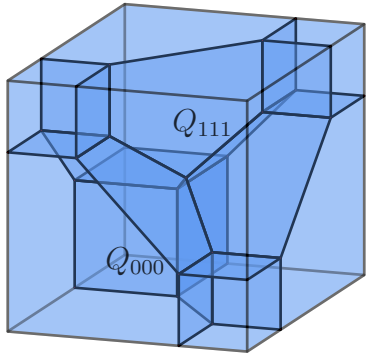

\begin{center}
  \include{counterexample}
\end{center}
\caption{Subdivision of the type space $[0,1]^3$ induced by the allocation function described in Example~\ref{ex:counter}.
  The region $Q_{000}$ corresponds to the corner in the lower back of the cube and $Q_{111}$ corresponds to the upper front corner.
  This and other pictures were obtained via \polymake~\cite{DMV:polymake}.}
\label{fig:counter}
\end{figure}

We define two allocation functions $f,g:\Theta\to\Omega$ as \emph{combinatorially equivalent} if their indifference complexes agree; i.e., $\cI(f)=\cI(g)$.
As before we are primarily concerned with the case where $\Omega=\{0,1\}^m$ and the allocation functions are truthful allocations of $m$ items.
In this way we can relate the allocation space with regular subdivisions of the cube $[0,1]^m$.
\begin{definition}
  A truthful allocation function on $m$ items is \emph{nondegenerate} if the associated regular subdivision of the $m$-cube is a triangulation.
\end{definition}

\ifipco
\begin{table}[tb]
\else
\begin{table}[b]
\fi
  \caption{Triangulations of $m$-cubes. Orbit sizes refer to regular triangulations}
  \label{tab:triangulations}
  \renewcommand{\arraystretch}{0.9}
  \begin{tabular*}{\linewidth}{@{\extracolsep{\fill}}rrrrr@{}}\toprule
    $m$ & all & regular & $\Sym(m)$-orbits & $\Gamma_m$-orbits \\
    \midrule
    2 &  2 &  2 &  2 & 1 \\
    3 & 74 & 74 & 23 & 6 \\
    4 & 92{,}487{,}256 & 87{,}959{,}448 & 3{,}706{,}261 & 235{,}277 \\
    \bottomrule
  \end{tabular*}
\end{table}

Triangulations of $m$-cubes are described in \cite[\S6.3]{Triangulations}.
The first two columns of Table~\ref{tab:triangulations} summarize the known values of the number of all (regular) triangulations of the $m$-cube.
In particular the second column shows the number of combinatorial types of nondegenerate truthful allocations.
The number of all, not necessarily regular, triangulations of the $4$-cube was found by Pournin \cite{Pournin:2013}.
The corresponding numbers of triangulations for $m\geq 5$ are unknown.

\begin{remark}\label{rem:regularRefinement}
 Any regular subdivision may be refined to a regular triangulation, on the same set of vertices; see \cite[Lemma 2.3.15]{Triangulations}.
\end{remark}

Our next goal is to explain the third and fourth columns of Table~\ref{tab:triangulations}.
To this end we need to discuss the symmetries of the cube, which are known.
That will be the key to understanding (truthful) allocations of exchangeable items.
The automorphism group, $\Gamma_m$, of the $m$-cube $[0,1]^m$ comprises those bijections on the vertex set which map faces to faces.
The group $\Gamma_m$ is known to be a semidirect product of the symmetric group $\Sym(n)$ with $\ZZ_2^m$; its order is $m!\cdot 2^m$.
Here the $j$-th component of $\ZZ_2^{m}$ flips the $j$-th coordinate, and this is a reflection at the affine hyperplane $x_j=\tfrac{1}{2}$; that map does not have any fixed points among the vertices of $[0,1]^m$.
The subgroup $\ZZ_2^m$ of all coordinate flips acts transitively on the $2^m$ vertices.
The symmetric group $\Sym(m)$ naturally acts on the coordinate directions; this is precisely the stabilizer of the origin in $\Gamma_m$; it acts transitively on the set of vertices with, say, $k$ ones and $m-k$ zeros.
Since the cells in each triangulation of $[0,1]^m$ are convex hulls of a subset of the vertices, the group $\Gamma_m$ also acts on the set of all triangulations of $[0,1]^m$.
Moreover, since $\Gamma_m$ acts via affine maps, it sends regular triangulations to regular triangulations.

The stabilizer $\Sym(m)$ acts transitively on the $\tbinom{m}{k}$ vertices of $[0,1]^m$ with exactly $k$ ones.
In this way, a $\Sym(m)$-orbit of regular triangulations corresponds to a set of nondegenerate truthful allocations functions for which the indifference complexes agree, up to permuting the items.
We call such allocation functions \emph{combinatorially equivalent for exchangeable items}.
The $\Sym(m)$-orbits of regular triangulations have been computed with \mptopcom \cite{JordanJoswigKastner:2018}; see the third column of Table~\ref{tab:triangulations}.
That computation furnishes a proof of the following result.

\begin{theorem} \label{thm:Cubes}
 There are 23 combinatorial types of nondegenerate truthful allocation functions for $n=3$ exchangeable items.
 Further, the corresponding count for $n=4$ yields 3{,}706{,}261.
\end{theorem}

\begin{figure}[tb]
  \input{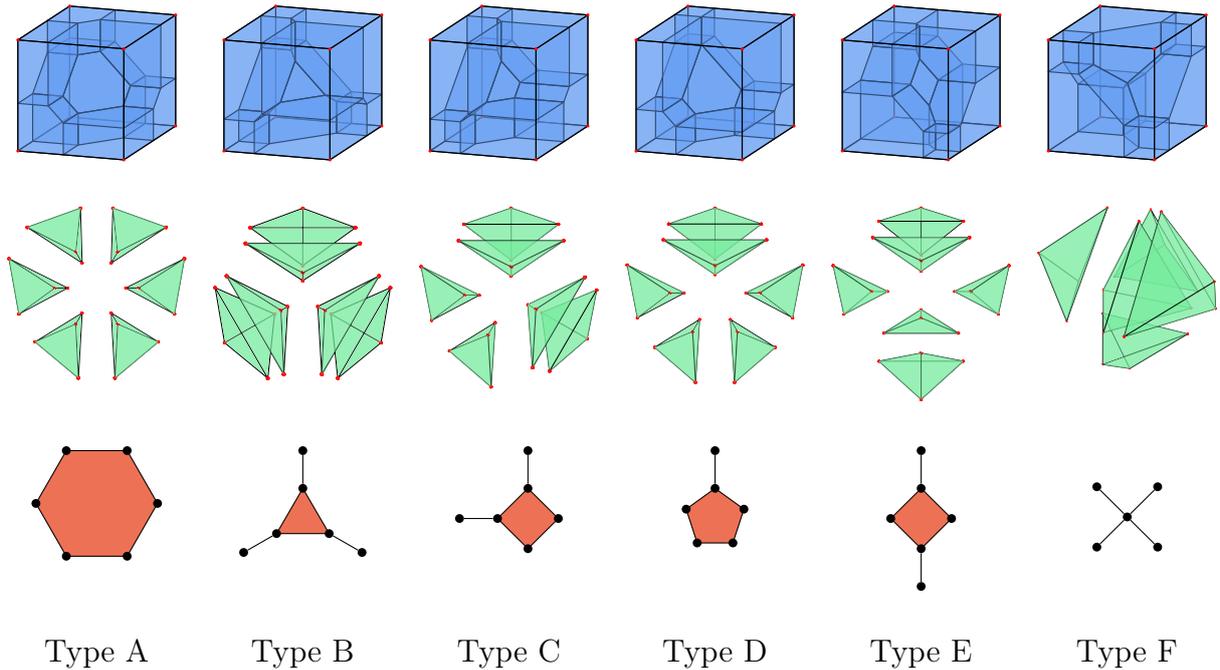}
  \caption{The types of truthful allocations corresponding to the six $\Gamma_3$-orbits of the $3$-cube, together with the corresponding triangulations (exploded) and their tight spans; cf.\ \cite[Fig.~6.35]{Triangulations}.
    The \emph{tight span} of a regular triangulation $\cS$ is the subcomplex of bounded cells of the tropical hypersurface dual to $\cS$ (seen as an ordinary polyhedral complex); see \cite[\S10.7]{ETC}.}
  \label{fig:AllocsAndTriangs}
 \end{figure}

It makes sense to focus on the combinatorics of triangulations, without paying attention to their interpretations for auctions.
This amounts to studying the orbits of the full group $\Gamma_m$ acting on the set of (regular) triangulations; see the fourth column of Table~\ref{tab:triangulations}.
The six $\Gamma_3$-orbits of triangulations of the $3$-cube are depicted in Fig.~\ref{fig:AllocsAndTriangs}.
This number expands to 23 if we consider the possible choices of locating the origin.
We illustrate the idea for the subdivision of the type space in Fig.~\ref{fig:counter}.
Its $\Gamma_3$-orbit splits into two $\Sym(3)$-orbits: one from putting the origin in one of the four cubes or in one of the four noncubical cells.

\begin{remark}\label{rem:vidali}
  Vidali considered a more restrictive notion of nondegeneracy of allocation functions \cite[Definition 8]{Vidali09}, and in \cite[Theorem 1]{Vidali09} she arrived at a classification of five types for three items.
  We point out that that number should be corrected to six, which is the count of $\Gamma_3$-orbits of regular triangulations of $[0,1]^3$ reported in Table~\ref{fig:AllocsAndTriangs}.
  The missing type is F (as in Fig.~\ref{fig:AllocsAndTriangs}), arising from Example~\ref{ex:counter}.
  The details are explained in the full version of this paper.
\end{remark}

\section{Sensitivity of Mechanisms}

In this section, we study by how much the allocations for a fixed player change under a slight modification of the reported type.
These changes are measured in the following two ways.
For two local allocations $a, b \in \{0,1\}^m$, let the \emph{cardinality distance} be
\ifipco
\else
\[
  \Cd(a,b) \ = \ \Big\vert \vert a \vert_1 - \vert b \vert_1 \Big\vert \enspace ,
\]
\fi
and let the \emph{Hamming distance} be
\ifipco
\else
\[
  \Hd(a,b) \ = \ \vert a - b \vert_1 \enspace ,
\]
\fi
where $\vert\cdot\vert_1$ is the $1$-norm.
Note that the cardinality distance is a pseudometric.
Let further $\Phi_m$ be the set of implementable indifference complexes on $m$ items.
Then, we define the \emph{cardinality sensitivity} as 
\ifipco
\else
\[
  \mu_c(m) \ = \ \min_{\cI \in \Phi_m} \biggl\{ \max \bigSetOf{ \Cd(a,b)}{ a,b \in F \text{ for some } F \in \cI} \biggr\} \enspace .
\]
\fi
The \emph{Hamming sensitivity} $\mu_h(m)$ arises in the same way, with $\Hd$ instead of $\Cd$.
Intuitively, the cardinality sensitivity $\mu_c(m)$ is the minimal amount such that there is a one-player DSIC mechanism for $m$ items with the property that any slight change in the type of the player does not cause her allocated bundle to change its cardinality by more than $\mu_c(m)$.
  
Our strategy to compute these values is as follows.
From Theorem~\ref{thm:TriangCube} we know that there is a bijection between $\Phi_m$ and the set of regular subdivisions of $[0,1]^m$.
So we need to identify those subdivisions, for which the maximal distance between any two vertices of one of its cells is minimized.
In this way, we can compute $\mu_c(m)$ exactly, and we give bounds for $\mu_h(m)$.

\begin{proposition}\label{prop:CardStab}
 The cardinality sensitivity of DSIC one-player auctions \ifipco\else is \fi $\mu_c(m) = 1$.
\end{proposition}
\begin{proof}
 We first slice the unit cube into the polytopes 
 \ifipco
\else
 \[
   P_k \ = \ \SetOf{x \in [0,1]^m}{k-1 \leq \sum_{i = 1}^m x_i \leq k} \; , \quad k = 1, \dots, m \enspace .
 \]
 \fi
 The polytopes $P_1, \dots, P_m$ form the maximal cells of a polytopal subdivision, $\cS$ of $[0,1]^m$.
 That subdivision is regular with height function $\lambda(x) = -\left(\sum_{i = 1}^m x_i\right)^2$.
 This proves the claim, as for each $P_k$, the difference in the coordinate sums of two of its vertices differ by at most one.
 \ifipco
 \else

 To show that $\cS$ is a regular subdivision, let $k$ be fixed and let $v$ be a vertex of the $m$-cube with $k+\delta$ many ones, where $\delta \in \{-k, \dots, m-k\}$.
 We observe that
 \begin{align*}
   \sum_{i \in [m]} v_i + \frac{\lambda(v)}{2k-1} \;
   &= \; \frac{1}{2k-1}\left( (k+\delta)(2k-1) - (k+\delta)^2 \right) \\
   &= \; \frac{1}{2k-1}(k+\delta)(k-\delta-1) \; = \; \frac{k^2-k-\delta(\delta+1)}{2k-1} \;
     \leq \; \frac{k^2-k}{2k-1} \enspace ,
 \end{align*}
 where the last inequality is tight only if $\delta \in \{-1,0\}$.
 Therefore, the hyperplane
 \[
   H_k \ = \ \BiggSetOf{(x,\lambda) \in \RR^m \times \RR}{ \sum_{i \in [m]} x_i + \frac{\lambda}{2k-1} = \frac{k^2 - k}{2k-1}}
 \]
 is a supporting hyperplane of the lifted unit cube $\conv \SetOf{(x,\lambda(x))}{x \in \{0,1\}^m}$, and the projection of the intersection of $H_k$ with the lifted unit cube is $P_k$.
 \fi
\end{proof}

Note that a mechanism corresponding to the indifference complex we used to prove the cardinality sensitivity of one can be obtained by choosing the prices $p(a) = \left(\sum_{i=1}^m a_i\right)^2$ for the allocations $a \in \{0,1\}^m$.

\begin{proposition}
\label{pro:hamming}
 The Hamming sensitivity for DSIC one-player auctions on $m \geq 3$ items is bounded by
 $2 \leq \mu_h(m) \leq m-1$.
\end{proposition}
\begin{proof}
 For the lower bound let us consider a triangle with the vertices $a,b,c \in \{0,1\}^m$.
 If we assume $\Hd(a,b) = \Hd(a,c) = 1$ then the vertices $a$ and $b$ (resp.\ $a$ and $c$) differ by a coordinate flip.
 Therefore, the vertices $b$ and $c$ differ by either two coordinate flips or none.
 As $b \neq c$, the former is the case and $\Hd(b,c) = 2$.
 As the maximal cells of a subdivision of $[0,1]^m$ for $m\geq 2$ contain at least three vertices, this proves the lower bound.
 
 For the upper bound, we show that there is a subdivision, $\cS$, of $[0,1]^m$ such that no cell of $\cS$ contains two antipodal vertices, i.e., two vertices such that their sum equals the all ones vector.
 We first consider the case where $m$ is odd.
 For a vertex $x\in\{0,1\}^m$, let $\Delta(x)$ be the \emph{cornered simplex} with apex $x$.
 That is, its vertices comprise $x$ and all its neighbors in the vertex-edge graph of the unit cube; cf.\ \cite[Fig.~6.3.1]{Triangulations}.
 Let $\cS_m$ be the subdivision of $[0,1]^m$ with the following maximal cells: the \emph{big cell} is the convex hull of all vertices with an even number of ones, and the \emph{small cells} are the cornered simplices $\Delta(x)$, where $x\in\{0,1\}^m$ with $\sum x_i$ odd.
 The subdivision $\cS_3$ is the triangulation of type F in Fig.~\ref{fig:AllocsAndTriangs}; for $m\geq 5$ the big cell is not a simplex, and so $\cS_m$ is not a triangulation in general.

 At any rate, the subdivision $\cS_m$ is always regular: it is induced by the height function which sends a vertex $x$ to $0$, if it has an even number of ones and to $-1$, if that number is odd.
 For $m\geq 3$ odd, no antipodal pair of vertices is adjacent in the vertex-edge graph of any cell of $\cS_m$.
 
 If the dimension $m$ is even, we consider the $m$-dimensional unit cube as a prism over $[0,1]^{m-1}$.
 Then $m-1$ is odd, and we can employ the subdivision $\cS_{m-1}$ of $[0,1]^{m-1}$ that we discussed before.
 We obtain a subdivision, $\cS_m$, of $[0,1]^m$ whose maximal cells are prisms over the maximal cells of $\cS_{m-1}$.
 The subdivision $\cS_m$ is again regular: this can be seen from assigning the vertices $x \times \{0\}$ and $x \times \{1\}$ the same height as the vertex $x$ in $\cS_{m-1}$.
 Now let $P$ be a maximal cell in $\cS_{m-1}$, such that $Q = P \times [0,1]$ is a maximal cell of $\cS_m$.
 If $Q$ contained an antipodal pair of vertices, then by removing the last coordinate, we would get an antipodal pair in $P$, which is absurd.
 The subdivision $\cS_m$ can be refined to a regular triangulation; see Remark~\ref{rem:regularRefinement}.
 This completes the proof.
\end{proof}

\section{Conclusion}

We studied DSIC allocation mechanisms where a set of $m$ items is allocated to $n$ players. These mechanisms can be described by the corresponding one-player mechanisms when the types declared by the other players are fixed.
For a single player, the allocations correspond to vectors $\{0,1\}^m$, and the combinatorial types of the allocation mechanisms correspond to regular subdivisions of the $m$-dimensional unit cube.
We then used this insight to design mechanisms that are robust in the sense that small changes in the declared type do not lead to a major change in the set of allocated items.
In the full version of this paper, we will show how this method can be applied in order to describe affine maximizers with $n$ players.

For multiple copies of items, the deterministic allocations to a single player correspond to a subset of the lattice $\NN^m$, and it seems plausible that DSIC mechanisms for such scenarios can also be described by regular subdivisions.
\begin{question}
  How does our approach generalize to allocation mechanisms in a setting with multiple copies of items?
\end{question}

\bibliographystyle{amsplain}
\bibliography{references.bib}

\end{document}